\documentclass[11pt]{article}
\usepackage{amsmath}
\usepackage{float}
\usepackage{graphicx}
\usepackage{graphics}
\usepackage[usenames]{color}
\usepackage{pdfpages}
\usepackage{vmargin}
\usepackage{subfigure}
\usepackage{hyperref}
\usepackage{ifthen}
\usepackage{listings} 
\usepackage{amssymb} 
\usepackage{amsthm}

\newcommand{\oen}{[\hspace*{-4.80pt}\ [} 
\newcommand{\cen}{]\hspace*{-4.80pt}\ ]} 
\newcommand{\oenP}{(\hspace*{-4.78pt}\ [} 
\newcommand{\cenP}{]\hspace*{-4.78pt}\ )} 
\newcommand{\notlongrightarrow}{\longrightarrow \hspace*{-15pt} /\ } 

\makeindex
\newtheorem{mydef}{Definition}
\newtheorem{mylemma}{Lemma}
\newtheorem{prope}{Property}
\newtheorem{myprop}{Proposition}
\newtheorem{example}{Example}

\begin{document}

\title{Towards a correct and efficient implementation of simulation and verification tools for probabilistic ntcc}
 \author{Mauricio Toro \\ Universidad Eafit \\ mtorobe@eafit.edu.co}
\bibliographystyle{alpha}
\date{}
\maketitle


\textbf{AMS Mathematics Subject CLassification} 68Q85, 68Q10, 68Q87

\section*{Abstract}
We extended our simulation tool \textit{Ntccrt} for probabilistic \texttt{ntcc} (\texttt{pntcc}) models. In addition, we developed a verification tool for \texttt{pntcc} models. Using this tool we can prove properties such as ``the system will go to a successful state with probability $p$ under $t$ discrete time-units''.

 Currently, we are facing a few problems. We can only verify \texttt{pntcc} models
 using a finite domain constraint system and the encoding of cells ( mathematical entities that can update their value ) is experimental. In addition, in order to
 reduce the states generated during the verification process we need to implement
 a procedure to calculate whether two processes are equivalent.

 In the future, we want to provide multiple interfaces for the tools (e.g., a web application, a graphical interface and command line interface). We also want
 to support constraint systems over trees, graph and sets. We want to show
 the relevance of our tool to model biological and multimedia

 interaction systems in our tool, verify some properties about them, and simulate such systems in our real-time capable interpreter.

Process calculi has been applied to the modeling of interactive music systems
 \cite{toro2018, toro2016, toro2016faust, toro2016gelisp, 2016arXiv160202169T, toro2015ntccrt, is-chapter,tdcr14,ntccrt,cc-chapter,torophd,torobsc,Toro-Bermudez10,Toro15,ArandaAOPRTV09,tdcc12,toro-report09,tdc10,tdcb10,tororeport} 
 and ecological systems \cite{EPTCS2047, PT13,TPSK14,PTA13}. In addition, research on algorithms \cite{PAT2016, MorenoPT17, RestrepoPT17}
 and software engineering \cite{toro2018a} also contributes to this field.

\section{Introduction}
\begin{mydef}
\label{ntccinternal}
\texttt{pntcc} internal transitions.
\end{mydef}

\begin{mylemma}
\label{terminating}
Every sequence of internal sequences is terminating (i.e., there are not infinite sequences) [\textit{ntcc-phd}].
\end{mylemma}

\begin{mydef}
Given a  configuration
$\langle P,c \rangle$, where $c$ is 
an initial store (i.e., input)  and $P$ is a \texttt{ntcc} process. We define a \textbf{output} as a store $d$ such that $\langle P,c \rangle \longrightarrow^* \langle Q,d \rangle \notlongrightarrow$, where $\longrightarrow$ means ... $\notlongrightarrow$ means ... $\longrightarrow^*$ means...
\end{mydef}

\begin{mydef}
A Constraint System (CS) is a pair ($\sum$,$\Delta$) where $\sum$ is a signature specifying constants, functions and predicate symbols, and $\Delta$ is consistent first-order they over $\sum$ (i.e., a set of first-order sentences over $\sum$ having at least one model). We say that $c$ entails $d$ in $\Delta$, written $c \models_\Delta d$ iff the formula $c \Rightarrow d$ is true in all models of $\Delta$. We write $\models$ instead of $\models_\Delta$ when $\Delta$ is unimportant [\textit{ntcc-phd}].
\end{mydef}

\begin{mydef}
Let $n > 0$. \textbf{FD(n)} is a CS such that:

- $\sum$ is given by constant symbols $0..n-1$ and the equality.

- $\Delta$ is given by ...
[\textit{ntcc-phd}]
\end{mydef}

\begin{mydef}
A Herbrand CS ...
\end{mydef}

\begin{mydef}
\label{csp}
A \textbf{CSP} is defined as a triple $\langle X, D, C \rangle$ where $X$ is a set of variables, $D$ is a set of domains and $C$ is a set of constraints. A solution for a CSP is an evaluation that satisfies all constraints.
\end{mydef}

The input-output behavior can be interpreted as an interaction between the system $P$ and the environment. At the time unit $i$, the environment provides a stimulus $c_i$ and $P_i$ produces $c_i'$ as a response. As observers, we can see that on input $\alpha$ the process P responds with $\alpha'$. We then recard ($\alpha, \alpha'$) as a reactive observation of P. Given $P$ we shall refer to the set of all its reactive observations as the input-output behavior of $P$ [\textit{ntcc-phd}].

\begin{mydef}
\label{iodef}
$io(P) = \{ (\alpha, \alpha' | P \overset{\alpha,\alpha'}{\Longrightarrow} ^\omega\}$
\end{mydef}

\noindent
We recall the following proposition from the \texttt{pntcc} paper.

\begin{myprop}
Given a pntcc process $P_0$, for every $P_n$ reachable from $P_0$ through an observable sequence, in the DTMC given by DTMC($\langle P_o, true \rangle)$ there exists a path from $\langle P_0, true \rangle$ to $\langle P_n, d \rangle$, for some constraint $d$.
\end{myprop}

\section{Simulation for \texttt{pntcc}}
In what follows, we explain how to formalize the construction of an interpreter for \texttt{pntcc} and its implementation.

\subsection{Encoding a deterministic, non-timed, non-probabilistic fragment of \texttt{pntcc} as a CSP}
In this section we propose the encoding of a \texttt{pntcc} fragment as a Constraint Satisfaction Problem (CSP).
First, we give some useful definitions. Then, we present the enconding of a \texttt{pntcc} process into a constraint. Following, we prove the correctness of the encoding. Finally, we propose the encoding for the execution of a  
\texttt{pntcc} process and an input (i.e., store) parametrized by a specific scheduler into a CSP. 

We prove that all the solutions of the CSP are
valuations for the output of a \texttt{pntcc} process for such scheduler. An advantage of representing the execution of a process as a CSP is that we can use any constraint solving tool to simulate the execution of a process. 


The following fragment of \texttt{pntcc} does not include temporal operators (i.e., next, !, unless and *), non-deterministic choice (i.e., $\sum$) nor probabilistic choice (i.e., $\bigoplus$).
It is parametrized by a Finite Domain constraint system, which is also parametrized by $2^{32}$, which
is the size of an integer on a 32 bits computer architecture. 




\begin{mydef}
\label{ntccfragment}
A fragment of \texttt{pntcc} parametrized by FD[$2^{32}$]  where
\end{mydef}

\noindent
$P, Q$ ::= $\sum \limits_{i \in I}$ \textbf{when} $c$ \textbf{do} $P$ $|$ $P \| Q$ $|$ \textbf{tell} ($c$) $|$ \textbf{local} $x$ $P$\\ 

\noindent
In order to define the CSP, we need to define its variables. For that reason, we provide the function vars(P), which returns all the
non-local variables used by a \texttt{pntcc} process.

\begin{mydef}
Let vars(P): `` \texttt{pntcc} process'' $\rightarrow$ ``set of variable names'' be recursively defined.
\end{mydef}


vars(\textbf{tell}(c))::= $Cvars(c)$, the variables contained in a constraint.

vars($P || Q$) ::= vars(P)$\cup$vars(Q)

vars($\sum \limits_{i \in I}$ \textbf{when} $c$ \textbf{do} $P$) ::= $\bigcup \limits_{i \in I}$ vars($P_i$)

vars(\textbf{local} x $P$)::= vars(P)$- \{x\}$\\

\noindent
The encoding \oen.\cen codifies a \texttt{pntcc} process into a constraint. 
A key issue for this encoding is representing the non-deterministic process  $\sum \limits_{i \in I}$ \textbf{when} $c_i$ \textbf{do} $P_i$.


We propose the constraint $(c_i \leftrightarrow e_i \wedge f_i \leftrightarrow \oen P_i \cen \wedge (e_i \wedge f_i) \leftrightarrow g_i \wedge \sum \limits_{i \in I} g_i = 1) \vee \bigwedge \limits_{i \in I} \neg c_i$. The idea is posting the constraints associated to a process for at most one process
which guard holds.

We assume
a constraint $c_i \leftrightarrow b$ in the constraint system for each constraint used in the
``when'' processes. These constraints are called \textit{reified constraints}.
The use of constraints as guards of ``when'' processes will be limited by the reified constraints
supplied by the constraint solving tool.

\begin{mydef}
\label{encoding}
Let \oen.\cen:  ``\texttt{pntcc} process'' $\rightarrow$ ``FD constraint'' be defined recursively.
\end{mydef}


\oen\textbf{tell}($c$)\cen ::= $ c $, c is a FD constraint

\oen$P || Q$\cen ::= \oen P\cen$\wedge$\oen Q\cen


\oen $\sum \limits_{i \in I}$ \textbf{when} $c_i$ \textbf{do} $P_i$\cen ::= 
$(c_i \leftrightarrow e_i \wedge f_i \leftrightarrow \oen P_i \cen \wedge (e_i \wedge f_i) \leftrightarrow g_i \wedge \sum \limits_{i \in I} g_i = 1) \vee \bigwedge \limits_{i \in I} \neg c_i$


\oen\textbf{local} x $P$\cen ::= $\exists x. $\oen P\cen\\




\noindent
Using the encoding presented above, we show that for every constraint $c$,  given a process $P$, 
its output is equivalent to the encoding of $P$ in conjuction with the constraint $c$. This proposition makes a link between the output of a \texttt{pntcc} process and a constraint.

\begin{myprop}
\label{propositionencoding}
Let $P$ and $\oen .\cen$  be a process in the \texttt{pntcc} fragment given by Def. \ref{ntccfragment} and the encoding given by Def. \ref{encoding}.
Then, for every constraint $c$ using a scheduler  that chooses the process with minium index which guard holds, it holds that\\
$( \langle P,c \rangle \longrightarrow_{} ^* \langle Q,d \rangle \notlongrightarrow\ \  ) \rightarrow (\oen P \cen\wedge c \equiv d)$  
\end{myprop}

\begin{proof}
The proof proceeds by induction on the structure of $P$. 

\begin{enumerate}





\item 

$P$ = \textbf{tell}($c$). 

According to the rule TELL (Def. \ref{ntccinternal}) and the fact that \textbf{skip} does not make any internal transition, we must have
$\langle$ \textbf{tell}(c),d $\rangle$ $\longrightarrow_{}$ $\langle \textbf{skip}, d \wedge c \rangle \notlongrightarrow$


Then, we have to prove that $d \wedge c \equiv \oen P \cen \wedge d$ 

Since $\oen P \cen = c$ accoding to Def. \ref{encoding}, we have $d \wedge c \equiv c \wedge d$








\item 

$P$ = $Q \| R$

We recall the rule 

PAR $\frac{\langle Q,c \rangle \longrightarrow_{} \langle Q',d \rangle}{\langle Q \| R, c \rangle \longrightarrow_{} \langle Q' \| R,d \rangle}$

We suppose $\langle Q ,c \rangle \longrightarrow_{} \langle Q', e \rangle$ and $\langle R ,e \rangle \longrightarrow_{} \langle R', f \rangle$ (1)

From (1) and lemma \ref{terminating},  we can deduce
$\langle Q,c \rangle \longrightarrow_{} \langle Q', e \rangle \longrightarrow_{}^* \langle H, w \rangle \notlongrightarrow$

According to the rule PAR, to (1) and to the lemma \ref{terminating}, we must have 

$\langle Q \| R,c \rangle \longrightarrow_{} \langle Q' \| R, e \rangle \longrightarrow_{} \langle Q' \| R' , f \rangle \longrightarrow_{}^* \langle S \| R \rangle \notlongrightarrow$  (2)

by (2) and the inductive hypothesis, we must have

$f \equiv \oen Q \cen \wedge c$ (3)\ \ \ $f \equiv \oen Q' \cen \wedge e$ (4)

from (3) and (4), we have  $\oen Q \cen \wedge c \equiv \oen Q' \cen \wedge e$ (5)

On the other hand, from (2) and the inductive hypothesis, we can deduce

$g \equiv \oen Q \| R \cen \wedge c$\ \ \ (6) $g \equiv \oen Q' \| R \cen \wedge  e$ (7)

From (6) and (7), we deduce
$\oen Q \| R \cen \wedge c \equiv \oen Q' \| R \cen \wedge  e$ (8)

Applying Def. \ref{encoding}, we can deduce from (8)
$\oen Q \cen \wedge \oen R \cen \wedge c \equiv \oen Q' \cen \wedge \oen R \cen \wedge e$ (9)

Finally, replacing (5) in (9) we have
$\oen Q \cen \wedge \oen R \cen \wedge c \equiv \oen Q \cen \wedge \oen R \cen \wedge c$






The case where $\langle Q \| R',c \rangle \longrightarrow_{} \langle Q \| R', m \rangle \longrightarrow_{} \langle Q' \| R' , f \rangle \longrightarrow_{}^* \langle S \| R \rangle \notlongrightarrow$ is trivial since $Q \| R \equiv R \| Q$
according to structural congruence.

\item P = $\sum \limits_{i \in I} $ \textbf{when} $c_i$ \textbf{do} $P_i$

We recall the rule 

SUM $\frac{}{\langle \sum \limits_{i \in I} when\ c_i\ do\ P_i, d \rangle \rightarrow
  <P_i',d>}$ if $d \equiv c_i$ \\


There are two cases:

\begin{enumerate}

\item $\exists c_i . d \equiv c_i$

By lemma \ref{terminating} and SUM,  we must have
$\langle P, d \rangle \longrightarrow_{} \langle P_i, d \rangle \longrightarrow_{}^* \langle Q, e \rangle \notlongrightarrow $

Then, we have to prove 
$e \equiv d \wedge \oen P \cen$ (1)\\
We have $e \equiv d \wedge \oen P' \cen$ (2) by the inductive hypothesis

From (1) and (2) we deduce
$ d \wedge \oen P \cen \equiv d \wedge \oen P_i \cen$

Since $\oen P \cen = \oen P_i \cen$ by to Def. \ref{encoding} and because ???, we  have $d \wedge \oen P \cen \equiv d \wedge \oen P \cen$

\item $\neg  \exists c_i . d \models c_i$

We must have
$\langle P, d \rangle \notlongrightarrow$ 
 .Then, we have to prove that
$d \equiv \oen P \cen \wedge d$.

Since $\oen P \cen = true$ according to Def. \ref{encoding} and because $\bigwedge \limits_{i \in I} \neg ci$, we have $d \equiv d$

\end{enumerate}



\item $P$ = \textbf{local} $x$ $Q$

We recall the rule 

LOC $\frac{\langle Q,c \wedge \exists x d \rangle \rightarrow
  \langle Q',c \rangle}{\langle (local x,c) Q,d \rangle \rightarrow \langle (local x, c') Q', d \wedge
  \exists x c' \rangle}$. \\


According to lemma \ref{terminating} and LOC, we must have 

$\langle (local x, c) Q, d \rangle \longrightarrow_{} \langle (local x, c') Q', d \wedge \exists x c' \rangle \longrightarrow_{}^* \langle R,e \rangle \notlongrightarrow$

We have to prove 
$e \equiv \exists x. \oen Q' \cen \wedge d \wedge \exists x. c' $ (1)\ \ \ 
$e \equiv \exists x. \oen Q \cen \wedge d$ (2)

From (1) and (2) we have
$\exists x. \oen Q' \cen \wedge \exists x. c'  \equiv \exists \oen Q \cen $ (3)

To prove (3), we have to prove all these cases:
\begin{enumerate}

\item Q = \textbf{tell}(h)

We have Q' = \textbf{skip} and $c' = h$

Then, according to Def. \ref{encoding} and replacing Q, Q' and c in (3), we have\\
$\exists x. true \wedge \exists x. h \equiv \exists x. h$

\item Q = $\sum \limits_{i \in I}$ \textbf{when} $c_i$ \textbf{do} $Q_i$ 

Since $\oen Q \cen = \oen Q_k \cen$ by Def. \ref{encoding} and $c' = true$ by rule SUM and because we know that asume that $\langle Q, c \wedge \exists d \rangle \longrightarrow_{} \langle Q', c \rangle$, we have $\exists x. \oen Q \cen \equiv \exists x. \oen Q \cen$

\item Q = $S \| R$

According to rule PAR, we must have $Q' = S' \| R$ and $c' = d$ 

Then, we have to prove that
$\exists x (\oen S' \| R \cen) \wedge \exists x. d \equiv \exists x \oen Q \cen$

We have the following derivation

$\exists x (\oen S' \| R \cen) \wedge \exists x. d  \equiv \exists x ( \oen S' \cen \wedge \oen R \cen) \wedge \exists x . d$ by Def. \ref{encoding}

$\exists x ( \oen S' \cen \wedge \oen R \cen) \wedge \exists x . d \equiv \exists x (\oen S' \cen \wedge d \wedge \oen R\cen)$

$\exists x (\oen S' \cen \wedge d \wedge \oen R\cen) \equiv \exists x(\oen S' \cen \wedge \oen R \cen )$ by inductive hypothesis (recalling proof for the case $P = Q \| R$)

$\exists x \oen S' \cen \wedge \oen R \cen \equiv \exists x \oen Q\cen  $ by Def. \ref{encoding}

\item $Q$ = \textbf{local} $x$ $Q'$

This holds by the inductive hypothesis

\end{enumerate}

\end{enumerate}
\end{proof}

\noindent
Since we can encode the output of a process as a constraint, following, we describe the
relation between executing a process and solving a CSP. We will show that a process $P$
and a constraint $c$ can be rewriten as a CSP, and the solutions of such CSP are all the valuations of the output obtained by executing a process with the scheduler described previously.

\begin{myprop}
Let $P$ be a process in the \texttt{pntcc} fragment given by Def. \ref{ntccfragment} and
$\oen .\cen$ the encoding given in Def. \ref{encoding}
, for every $c$ it holds that the solutions to the CSP 

\begin{itemize}
\item
Variables = $vars(P)$
\item
Domains = $[0..2^{32}-1]$ for each variable
\item
Constraints = $\{c \wedge \oen P \cen\}$
\end{itemize}

\noindent
are
all the valuations for a store $d$ obtained by executing $P$ with an input $c$ described in proposition \ref{}. Formally,\\
$\langle P,c \rangle \longrightarrow_{}^* \langle Q, d \rangle \notlongrightarrow $ $ \rightarrow$ $\forall x.(x = solution(CSP) \leftrightarrow d[x])$
\end{myprop}

\begin{proof}
According to proposition \ref{propositionencoding}, it holds that the store $d$ is equivalent to the constraint \oen P \cen$\wedge c$  after executing a process $P$ with a store $c$. By Def. \ref{csp}, we know that a solution for a CSP satisfies all its constraints, thus, all the solutions of a CSP are all the possible valuations that satisfies its constraints. Therefore, all the possible solutions for the CSP, satisifies the store $d$.
\end{proof}

\noindent
The correctness of our tool will be based on solving a CSP correctly. Fortunately, there are multiple techniques and theories about how to solve a CSP composed by FD constraints. In the implementation, we will left the problem of solving a CSP to a constraint solving library called Gecode \cite{fastprop}. In order to execute the $i-th$ process instead of the first process, we define a pre-encoding for the sum process 

\oen $\sum \limits_{i \in I}$ \textbf{when} $c_i$ \textbf{do} $P_i$  \cen$_{R}$ ::=
$\sum \limits_{i \in R(I)}$ \textbf{when} $c_i$ \textbf{do} \oen$P_i$\cen$_{R}$\\ where $R :$ ``set of indexes'' $\rightarrow$ ``set of indexes'' change the other of the indexes.






\subsection{Adding time}
In this section we explain how we can extend the encoding proposed in Def. \ref{encoding} for the ``time'' operators. Then, we propose an abstract
machine capable of simulating a finite number of \texttt{pntcc} time-units. Finally, we will prove that there is a relation
between the execution of the abstract machine and the execution of \texttt{pntcc} process. 

\begin{mydef}
\label{timedpntccfragment}
A non-probabilistic fragment of \texttt{pntcc} parametrized by FD[$2^{32}$] where
\end{mydef}

\noindent
$P, Q$ ::= $\sum \limits_{i \in I}$ \textbf{when} $c_i$ \textbf{do} $P_i$ $|$ $P \| Q$ $|$ \textbf{tell} ($c$) $|$ \textbf{local} $x$ $P$
 $|$ \textbf{next} $P$ $|$ \textbf{unless} $c$ \textbf{next} $P$ $|$ $\textbf{!}P$\\








\noindent
The following encoding is a function that takes a \texttt{pntcc} process as given by Def. \ref{timedpntccfragment} and returns
a pair composed by the constraint associated to that process (based on Def. \ref{encoding}) and the process to be executed in the next time-unit (based on the definition of F(P)).

\begin{mydef}
\label{encodingpntccM}
Let $\oen.\cen_T$:  ``\texttt{pntcc} process'' $\rightarrow$ ``pair'' be defined recursively.
\end{mydef}


\oen\textbf{tell}($c$)$\cen_T$ ::= $(  c, \textbf{skip} )$, c is a FD constraint

\oen$P || Q$$\cen_T$::= $(c \wedge d, R\| S )$, where $(c, R) = \oen P\cen_T\ and\ (d, S) = \oen Q\cen_T$


\oen $\sum \limits_{i \in I}$ \textbf{when} $c_i$ \textbf{do} $P_i$$\cen_T$::= 
$(c_i \leftrightarrow e_i \wedge f_i \leftrightarrow \oen P_i \cen \wedge (e_i \wedge f_i) \leftrightarrow g_i \wedge \sum \limits_{i \in I} g_i = 1) \vee \bigwedge \limits_{i \in I} \neg c_i$


\oen\textbf{local} x $P$$\cen_T$ ::= $( \exists x. c, \textbf{local}\ x\  Q )$ , where $( c,Q) =  \oen P\cen_T $\\

\oen\textbf{next} $P$$\cen_T$ ::= $( \texttt{true}, P )$

\oen\textbf{unless} $c$ \textbf{next} $P$$\cen_T$ ::= 
$$ 
               \begin{cases}
                        ( \texttt{true}, \textbf{skip} ) &, c \leftrightarrow b \wedge b = true \\
                        ( \texttt{true}, P ) &, otherwise 
               \end{cases}
       $$

\oen\textbf{!}$P$$\cen_T$ ::= $( c, \textbf{!}P \|Q )$ , where $(c,Q) =  \oen P\cen_T $\\

\begin{example}
\label{parallelexample}
Let $P$ = \textbf{tell}(c) $\|$ \textbf{next} \textbf{tell} (d) $\|$ !(\textbf{tell}(e)$\|$\textbf{tell}(f)). \\
Then, $\oen P \cen_T$ = $\langle c \wedge e \wedge f, \textbf{tell}(d) \| \textbf{!}(\textbf{tell}(e)\|\textbf{tell}(f)) \rangle$
\end{example}

\begin{myprop}
\label{propencodingt}
Let $P$ be a process given by Def. \ref{timedpntccfragment}, for any constraint $c$, it holds for every FD constraint $c$ that \\
$\langle P,c \rangle \longrightarrow_{}^* \langle Q, d \rangle \rightarrow (d \equiv e \wedge c \ and\ F(Q) \equiv R)$ \\ where  $F(Q)$ is the future function applied to Q and $(e,R) = \oen P \cen_T$ the encoding given by Def. \ref{encodingpntccM} applied to $P$
\end{myprop}

\begin{proof}

... Pending
\end{proof}

In order to execute $P$, the pntccM machine first need to encode $P$ into a suitable machine term. A machine term $V$ is a triple composed by a FD constraint, a process and an integer.

\begin{mydef}
Syntax of pntccM. 
\end{mydef}

V ::= $\langle c,Q,j \rangle$, where

c is a  Finite Domain constraint

Q is a process defined in Def. \ref{timedpntccfragment}

$j > 0$\\

\noindent
The following function is used to encode a \texttt{pntcc} process into a pntccM term for a simulation of $n$ time-units.

\begin{mydef}
\label{encodingmachineterm}
Encoding a \texttt{pntcc} process into a pntccM term.

$\oenP P \cenP_{n,I} = \langle c \wedge I,Q, n \rangle, $ where $(c,Q) = \oen P \cen_T$

\end{mydef}

Once a process has been encoded to a machine term using Def. \ref{encodingmachineterm}, it can be executed by the machine. A given
term can be executed if $i > 0$. The reduction depends on the input
from the environment. The new machine term is formed by the output of
the process, the future function applied to the process and $i-1$.

\begin{mydef}
Reduction in pntccM for an input $I$.

\end{mydef}

$\langle c, P,i  \rangle \longrightarrow_{}^I \langle d \wedge I,Q, i-1 \rangle$, where $ (d,Q) = \oen P \cen$ and $i > 0$ 

$\langle c, P,i  \rangle \notlongrightarrow^I\ \ $, when $i \leq 0$\\

\noindent
Next, we define a finite simulation of the pntccM machine.

\begin{mydef}
Let $n$ be the number of time-units to simulate, $P$ a pntcc process defined in Def. \ref{timedpntccfragment}, and $I$ a sequence of $n$ inputs (FD constraints). A simulation $S_{P,n,I}$ is a sequence $c_1...c_n$ such that $\langle c_1, Q_1, n \rangle \longrightarrow^{I_2} \langle c_2, Q_2, n-1 \rangle ... \longrightarrow^{I_{n-1}} \langle c_n, Q_n,0\rangle \notlongrightarrow$\ \ and $\oenP P \cenP_{n,I_1} = \langle c_1, Q_1, n \rangle$
\end{mydef}

\begin{example}
Let $P$ be the process defined in example \ref{parallelexample}.
Then, a simulation $S_{P,5,b^5} = [c \wedge b \wedge e \wedge f, b \wedge d \wedge e \wedge f, b \wedge e \wedge f, b \wedge e \wedge f, b \wedge e \wedge f ]$
\end{example}

There is a relation between the input-output behavior of a \texttt{pntcc} process and a simulation of the pntccM machine.

\begin{myprop}
\label{simprop}
Let $S$ be a simulation parametrized by  a process $P$ given by Def. \ref{timedpntccfragment}, an integer $n$ and a sequence of FD constraints $I$.
$io(P) [1..n]$ (i.e., the first $n$ elements of the sequence) is equal to $S_{P,n,I}$
\end{myprop}

Proving proposition \ref{simprop} we show that the pntccM gives the same output as a process because an input-output sequence generated
by a process is equal to a simulation of the machine. Then, every output given by a machine is calculated by a process and viceversa.

\begin{proof}
The proof proceeds by induction over the io(P)[1..n] and the simulation $S$ sequence. 
Let $\alpha_1...\alpha_n = io(P)$ and $c_1...c_n = S$.

Base case:

$\alpha_1 = c_1$
is a collorary of proposition \ref{propencodingt} and Def. \ref{encodingmachineterm}.

Inductive case:

Let $\alpha_i = c_i$ be the inductive hypothesis.
We must prove that $\alpha_{i+1} = c_{i+1}$.
This is also a collorary of proposition \ref{propencodingt}.

\end{proof}

\subsection{Adding probabilistic choice}
In this section we will show how to encode a non-timed fragment of \texttt{pntcc} as a sequence of Propagation Problems (PP). The advantage of this approach is that we can make a implementation of this fragment of \texttt{pntcc} only using a constraint solving tool based on propagators and a random-number generation library. We will also show 
the correctness of the encodings as usual.

\begin{mydef}
Process up to level j. Let  PUL$_j$ : ``\texttt{pntcc} process'' $\rightarrow$ ``\texttt{pntcc} process''  be defined recursively.
\end{mydef}

PUL$_j$(\textbf{tell}(c)) ::= \textbf{tell}(c)

PUL$_j$($P || Q$)::= PUL$_j$($P$)$\|$PUL$_j$($Q$)

PUL$_j$($\sum \limits_{i \in I}$ \textbf{when} $c_i$ \textbf{do} $P_i$) ::=$\sum \limits_{i \in I}$ \textbf{when} $c_i$ \textbf{do} PUL$_j(P_i)$

PUL$_j$(\textbf{local} $x$ $P$) ::= \textbf{local} $x$ PUL$_j(P)$

PUL$_j$($\bigoplus \limits_{i \in I}$ \textbf{when} $c_i$ \textbf{do} $P_i$, $a_i$) ::= 
 $$ 
               \begin{cases}
                        \bigoplus \limits_{i \in I}\ \textbf{when}\ c_i\ \textbf{do}\ PUL_{j-1}(P_i), a_i &, j \geq 0\\
                        skip&, j  < 0 \\
               \end{cases}
       $$

\begin{mydef}
Probabilistic processes at level j. Let  PPAL$_j$ : ``\texttt{pntcc} process'' $\rightarrow$ ``set of \texttt{pntcc} process''  be defined recursively.
\end{mydef}

PPAL$_j$(\textbf{tell}(c)) ::= $\emptyset$

PPAL$_j$($P || Q$)::= PPAL$_j$($P$)$\cup$PPAL$_j$($Q$)

PPAL$_j$($\sum \limits_{i \in I}$ \textbf{when} $c_i$ \textbf{do} $P_i$) ::=$\bigcup \limits_{i \ I}$ PPAL$_j(P_i)$

PPAL$_j$(\textbf{local} $x$ $P$) ::= PPAL$_j(P)$

PPAL$_j$($\bigoplus \limits_{i \in I}$ \textbf{when} $c_i$ \textbf{do} $P_i$, $a_i$) ::= 
 $$ 
               \begin{cases}
                        \bigcup \limits_{i \in I} PPAL_{j-1}(P_i) \cup P&, j \geq 0\\
                        \emptyset&, j  < 0 \\
               \end{cases}
       $$

\begin{prope}
\label{propprobabilistic}
For any $b$, It exists $c$ and $c'$, such that $\langle PUJ_j(P), b \rangle \longrightarrow \langle R ,c$ and $\langle P, b \rangle \longrightarrow \langle S ,c ' \rangle$ For each $P' \in PPAL_j(P)$, $\langle \oen P' \cen_P , c\rangle \longrightarrow \langle Q, c \rangle \rightarrow \langle P', c \rangle \longrightarrow \langle Q, d \rangle$
\end{prope}

\begin{mydef}
\label{pntccfragment}
A fragment of \texttt{pntcc} parametrized by FD[$2^{32}$] where
\end{mydef}

\noindent
$P, Q$ ::= $\sum \limits_{i \in I}$ \textbf{when} $c_i$ \textbf{do} $P_i$ $|$ $P \| Q$ $|$ \textbf{tell} ($c$) $|$ \textbf{local} $x$ $P$ $\|$ $\bigoplus \limits_{i \in I}$ \textbf{when} $c_i$ \textbf{do} $P_i, a_i$ \\
and $P,Q$ holds the property \ref{propprobabilistic}

\begin{mydef}
The \textbf{Maximum $\bigoplus$ nested depth (mnp)} function. Let mnp : \texttt{pntcc} $\rightarrow$ $\mathbb{N}$ be recursively defined.
\end{mydef}

mnp(\textbf{tell}($c$)) = 0

mnp($P \| Q$) = max( mnp($P$), mnp($Q$))

mnp(\textbf{local} $x$ $P$) = mnp($P$)

mnp($\sum \limits_{i \in I}$ \textbf{when} $c_i$ \textbf{do} $P_i$) = max($\{j, j = mnp(P_i) \wedge i \in I \}$)

mnp($\bigoplus \limits_{i \in I}$ \textbf{when} $c_i$ \textbf{do} $P_i$, $a_i$) = max($\{j, j = mnp(P_i) \wedge i \in I \}$) + 1








\begin{mydef}
The \textbf{Boolean Variables of a Probabilistic Process (ProbGuards)} function. Let ProbGuards : ``\texttt{pntcc} process'' $\rightarrow$ ``set of tuples $\langle$ process, index, boolean var$\rangle$'' be defined recursively
\end{mydef}

ProbGuards(\textbf{skip}) = ProbGuards(\textbf{tell}($c$) = $\emptyset$

ProbGuards($P \| Q$) =  ProbGuards($P$)$\cup$ProbGuards($Q$)

ProbGuards($\sum \limits_{i \in I}$ \textbf{when} $c_i$ \textbf{do} $P_i$) = $\bigcup \limits_{i \in I}$ProbGuards($P_i$)

ProbGuards(\textbf{local} $x$ $P$) = $\{$ProbGuards($P$)$\}$

ProbGuards($\bigoplus$ \textbf{when} $c_i$ \textbf{do} $P_i$, $a_i$)  =
$   \bigcup \limits_{i \in I} ProbGuards_{j-i}(P_i) \cup   \{\exists b. (P,i,b) \wedge c_i \leftrightarrow b  \wedge i \in I \}  $

\begin{mydef}
getBool(P,i,b) ...
\end{mydef}

\begin{mydef}
ChooseProb(set)
\end{mydef}

\begin{mydef}
The encoding \oen .\cen$_j$ 
\end{mydef}

\oen\textbf{skip}\cen$_j$ ::= true

\oen\textbf{tell}($c$)\cen$_j$ ::= $ c $, c is a FD constraint

\oen$P || Q$\cen$_j$ ::= \oen P\cen$\wedge$\oen Q\cen


\oen $\sum \limits_{i \in I}$ \textbf{when} $c_i$ \textbf{do} $P_i$\cen$_j$ ::=
 $$ 
               \begin{cases}
                        \oen P_k \cen_j &, k \in \{h,  h \in I \wedge c_j \leftrightarrow b \wedge b = true \} \\
                        true &, otherwise 
               \end{cases}
       $$

\oen\textbf{local} x $P$\cen$_j$ ::= $\exists x. $\oen P\cen$_j$

\oen$\bigoplus \limits_{i \in I}$ \textbf{when} $c_i$ \textbf{do} $P_i$, $a_i$\cen$_j$  =

 $$ 
               \begin{cases}
                  \bigwedge \limits_{i \in I} c_i \rightarrow getBool(P,i,B) &, j = 0 \\
                  ChooseProb \{(P_i, a_i) , getBool(P,i,B) = true \wedge i  \in I \} &, j = 1 \\
                   true  &, j > 0\\
               \end{cases}
       $$

\begin{myprop}
$\langle PUJ_j(P), c \rangle \longrightarrow^* \langle Q, d \rangle \notlongrightarrow \wedge \forall b, b \in guards \rightarrow ... c  \wedge \oen P \cen_{j,G} \equiv d\ and\ \exists g \in G c_i \leftrightarrow b$
\end{myprop}

\begin{myprop}
Copy from the other one...

Let $i$ such that $0 < i < mnp(P)$. After calculating mutual fixpoints for $mnp(P)+1$ propagation problems for
\end{myprop}

$PP_0$
\begin{itemize}
\item
Variables = \textit{vars$_i$(P)}$\cup$ $\{z , (x,y,z) =$ ProbGuards(P) $\}$
\item
Domains = 
$$
\begin{cases}
[0..2^{32}-1] &, var \in vars(PUL_0(P) \\
[0..1] &, var \notin vars(PUL_0(P) \\
\end{cases}
$$
\item
Constraints = $c \wedge \oen PUL_0(P) \cen_0$
\end{itemize}

Calculate a mutual fixpoint for all the propagators

Add the constraints = $\oen PUL_j(P) \cen_j$

and we think about the last PP as a CSP, the solutions to $vars_{mnp(P)+1}$ of the CSP ...

\subsection{Implementation}
\textit{Ntccrt} \cite{tororeport} is written in C++ and it can generate stand-alone programs executing a \texttt{ntcc} model. \textit{Ntccrt}  can also use Flext to generate the externals for either Max or Pd \cite{max}, and Gecode \cite{fastprop} for constraint solving and concurrency control.
Gecode is an efficient constraint solving library, providing efficient propagators (narrowing operators reducing the set of possible values for some variables). 
The basic principle of \textit{Ntccrt} is encoding the ``when'' and ``tell'' 
processes as Gecode propagators.

Although Gecode was designed to solve combinatorial problems, Toro
found out in  \cite{tororeport} that writing the ``when'' process as a
propagator, Gecode can manage all the
concurrency needed to represent \texttt{ntcc}.
Following, we explain the encoding of the ``tell'' and the ``when''.

To represent the ``tell'', we define a super class $Tell$. For \textit{Ntccrt}, we provide three subclasses to represent these processes: \textbf{tell} ($a=b$), \textbf{tell} ($a \in B$), and \textbf{tell} ($a > b$).
 Other kind of ``tells'' can be easily defined by inheriting
from the $Tell$ superclass and declaring an $Execute$ method.

We have a \textit{When propagator} for the ``when'' and a \textit{When} class for calling the propagator.
A process
\textbf{when} $C$ \textbf{do} $P$ is represented by two propagators:
$C \leftrightarrow b$ (a reified propagator for the constraint $C$) and \textbf{if} $b$ \textbf{then} $P$ \textbf{else} $skip$ (the \textit{When  propagator}).
The \textit{When propagator} checks the value of $b$. If the value of $b$ is true, it calls the \textit{Execute} method
of $P$. Otherwise, it does not take any action. Figure \ref{fig:whenprop} shows how to encode the process \textbf{when} $a=c$ \textbf{do} $P$ using our \textit{When propagator}.

\begin{figure}[!h]
  \begin{center}
{\includegraphics[width=0.48\columnwidth]{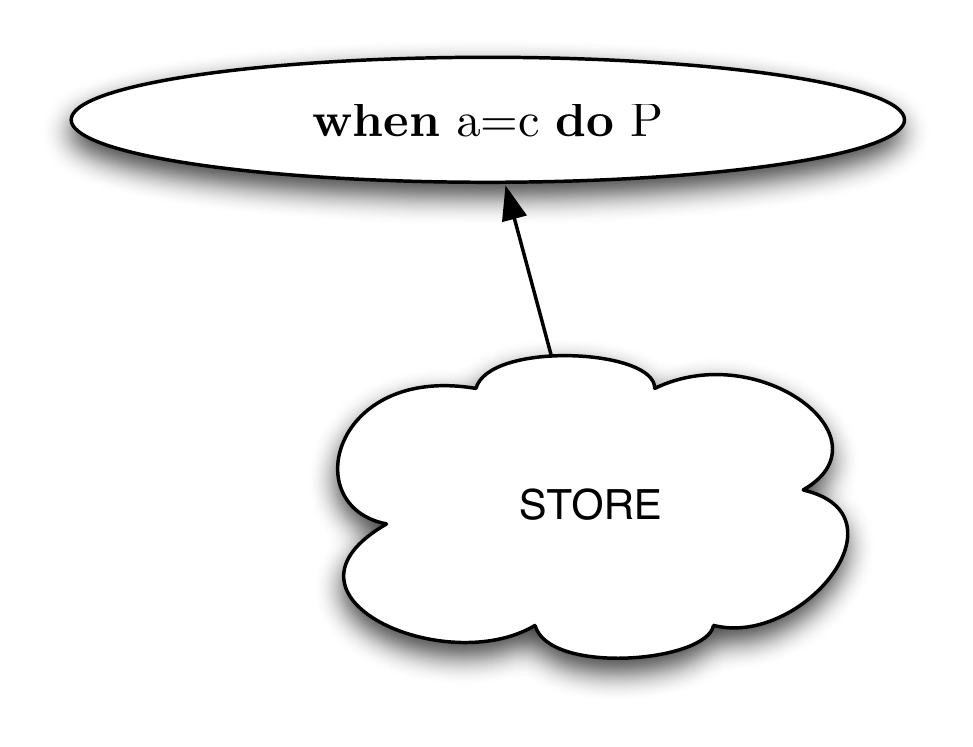}} \ \ 
{\includegraphics[width=0.48\columnwidth]{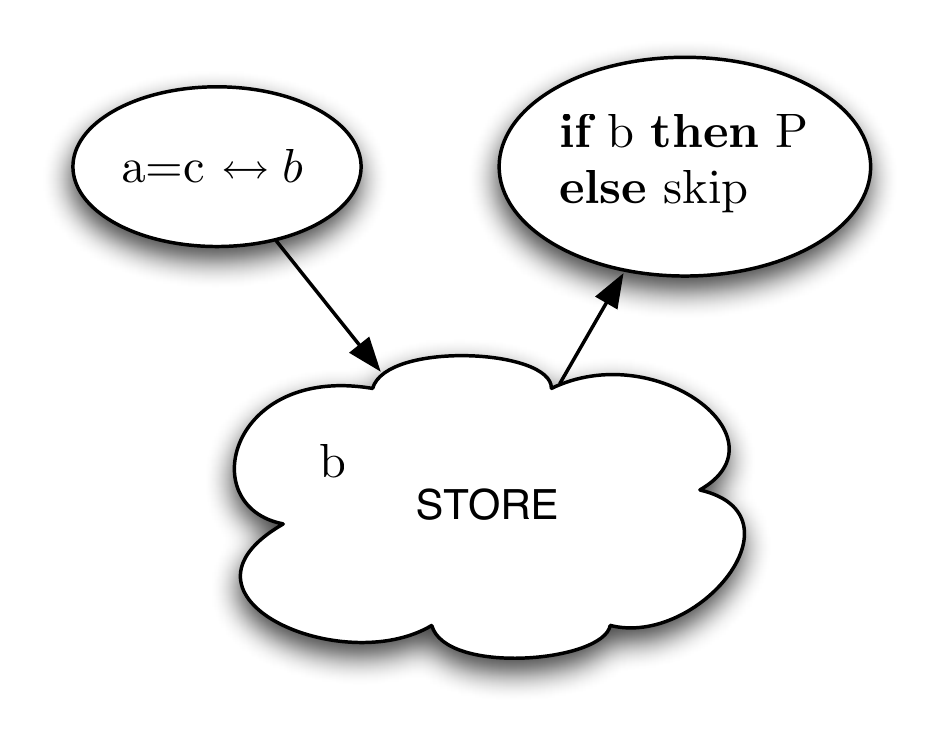}} 
    \caption{Example of the \textit{When propagator}}
    \label{fig:whenprop}
  \end{center}
\end{figure}

We have a \textit{$\sum$ propagator} ...

\section{Verification}
Finally, we will extend the abstract machine to calculate a Discrete-Time Markov Chain (DTMC) and we will show its correctness.
\subsubsection{Verification of non-probabilistic pntcc}

A key aspect of \texttt{pntcc} is that it can be used for both, simulation and verification. Following, we define another abstract machine that calculates DTMCs instead of a simple sequence of outputs. Using the DTMC we can prove PCTL properties. This machine is an extension of pntccM and it is also based on the encoding given by Def.  \ref{encodingpntccM}.

In order to execute $P$, the VerificationPntccM machine first need to encode $P$ into a suitable machine term. A machine term $V$ is a tuple composed by a FD constraint, a process, a DTMC and an integer.

\begin{mydef}
Syntax of VerificationPntccM. 
\end{mydef}

V ::= $\langle A,B,C,j \rangle$, where

A is a  Finite Domain constraint

B is a process defined in Def. \ref{timedpntccfragment}

C is a DTMC (i.e., a tuple $\langle Q_{OBS}, T_{OBS}, LM \rangle$)

$j > 0$

The following function is used to encode a \texttt{pntcc} process into a VerificationPntccM term to calculate a DTMC representing a $n$ time-units execution.

\begin{mydef}
\label{encodingmachinetermV}
Encoding a \texttt{pntcc} process into a VerificationPntccM term.

$\oenP P \cenP_{n,I} = \langle P_1 \wedge I, P_2, \langle n, \emptyset, (n, P_1) \rangle n \rangle, $ where $(P_1,P_2) = \oen P \cen_T$

\end{mydef}

\begin{mydef}
Reduction in VerificationPntccM.

\end{mydef}

$\langle A, B,C,i  \rangle \longrightarrow^I \langle B_1 \wedge I, B_2,\\ \langle Q_{OBS} \cup \{ i-1 \} , \gamma_0, T_{OBS} \cup \{(i, i-1,1.0)\}, LM \cup \{ (B,B_1 \wedge I)\}\rangle, i-1 \rangle$\\ where $( B_1 ,B_2) = \oen B \cen$, $\langle Q_{OBS}, \gamma_0, T_{OBS}, LM \rangle = C$ and $i > 0$ \\

$\langle A, B,C,i  \rangle \notlongrightarrow$\ \ , when $i \leq 0$

\begin{myprop}
Given a \texttt{pntcc} process $P$, the DTMC given by  the first ? states of the  DTMC($\langle P, true \rangle$) is an isomorph of the verification structure Ver(P,n)...
\end{myprop}

\begin{proof}
The proof proceeds by induction DTMC($\langle P, true \rangle$)
\end{proof}
\section{Applications}

\subsection{Herman's Stabilization protocol}
\subsection{Description}
A self-stabilising protocol for a network of processes is a protocol which, when started from some possibly illegal start configuration, returns to a legal/stable configuration without any outside intervention within some finite number of steps. For further details on self-stabilisation see here.

In each of the protocols we consider, the network is a ring of identical processes. The stable configurations are those where there is exactly one process designated as "privileged" (has a token). This privilege (token) should be passed around the ring forever in a fair manner.

For each of the protocols, we compute the minimum probability of reaching a stable configuration and the maximum expected time (number of steps) to reach a stable configuration (given that the above probability is 1) over every possible initial configuration of the protocol.

The first protocol we consider is due to Herman [Her90]. The protocol operates synchronously, the ring is oriented, and communication is unidirectional in the ring. In this protocol the number of processes in the ring must be odd.

Each process in the ring has a local boolean variable xi, and there is a token in place i if xi=x(i-1). In a basic step of the protocol, if the current values of xi and x(i-1) are equal, then it makes a (uniform) random choice as to the next value of xi, and otherwise it sets it equal to the current value of x(i-1).
\subsection{Simulation}

\begin{figure}[!h]
  \begin{center}
{\includegraphics[width=0.78\columnwidth]{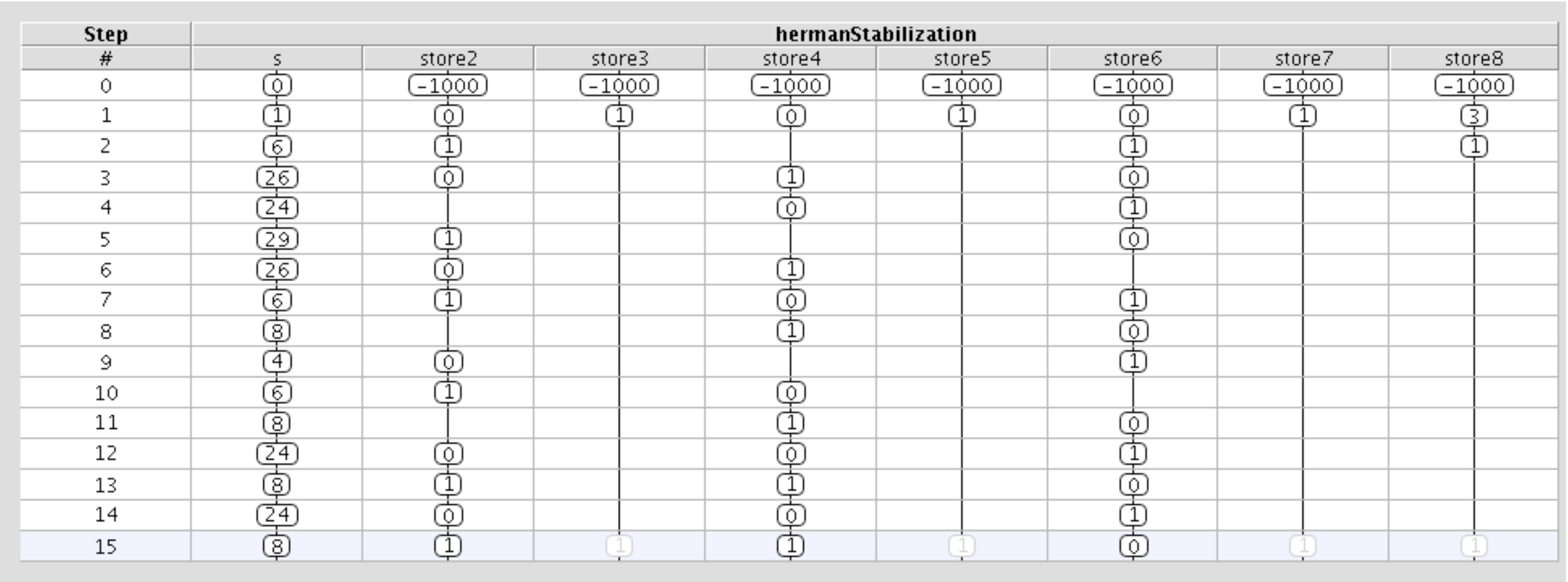}}  \\
 
    \caption{simulation}
    \label{fig:whenprop}
  \end{center}
\end{figure}
\begin{verbatim}
"Ntccrt Simulation" 
"" 
"" 
"Time Unit \# 0" 
"Number of processes= 7" 
"" 
"Variables Value (Only those specified are printed)" 
"num\_tokens" 
"3" 
"x1" 
"0" 
"x2" 
"0" 
"x3" 
"0" 
"changex1" 
"1" 
"changex2" 
"1" 
"changex3" 
"1" 
"..........................................." 
"" 
"Time Unit # 1" 
"Number of processes= 18" 
"" 
"Variables Value (Only those specified are printed)" 
"num_tokens" 
"1" 
"x1" 
"1" 
"x2" 
"0" 
"x3" 
"0" 

"..........................................." 
"" 
"Time Unit # 2" 
"Number of processes= 22" 
"" 
"Variables Value (Only those specified are printed)" 
"num_tokens" 
"1" 
"x1" 
"0" 
"x2" 
"1" 
"x3" 
"1" 
"changex1" 
"1" 
"changex2" 
"1" 
"changex3" 
"1" 
"..........................................." 
"" 
"Time Unit # 3" 
"Number of processes= 26" 
"" 
"Variables Value (Only those specified are printed)" 
"num_tokens" 
"1" 
"x1" 
"1" 
"x2" 
"0" 
"x3" 
"0" 
"changex1" 
"1" 
"changex2" 
"1" 
"changex3" 
"1" 
"..........................................." 
"" 
"Time Unit # 4" 
"Number of processes= 30" 
"" 
"Variables Value (Only those specified are printed)" 
"num_tokens" 
"1" 
"x1" 
"0" 
"x2" 
"1" 
"x3" 
"1" 
"changex1" 
"1" 
"changex2" 
"1" 
"changex3" 
"1" 
"..........................................." 
"" 
"Time Unit # 5" 
"Number of processes= 34" 
"" 
"Variables Value (Only those specified are printed)" 
"num_tokens" 
"1" 
"x1" 
"1" 
"x2" 
"0" 
"x3" 
"0" 
"changex1" 
"1" 
"changex2" 
"1" 
"changex3" 
"1" 
"..........................................." 
"" 
"Time Unit # 6" 
"Number of processes= 38" 
"" 
"Variables Value (Only those specified are printed)" 
"num_tokens" 
"1" 
"x1" 
"0" 
"x2" 
"1" 
"x3" 
"1" 
"changex1" 
"1" 
"changex2" 
"1" 
"changex3" 
"1" 
"..........................................." 
"" 
"Time Unit # 7" 
"Number of processes= 42" 
"" 
"Variables Value (Only those specified are printed)" 
"num_tokens" 
"1" 
"x1" 
"1" 
"x2" 
"0" 
"x3" 
"1" 
"changex1" 
"1" 
"changex2" 
"1" 
"changex3" 
"1" 
"..........................................." 
"" 
"Time Unit # 8" 
"Number of processes= 46" 
"" 
"Variables Value (Only those specified are printed)" 
"num_tokens" 
"1" 
"x1" 
"1" 
"x2" 
"1" 
"x3" 
"0" 
"changex1" 
"1" 
"changex2" 
"1" 
"changex3" 
"1" 
"..........................................." 
NIL 
T
\end{verbatim}
CL-USER 10 > 

\subsection{Model Checking}

We first check the correctness of the protocol, namely that:

From any configuration, a stable configuration is reached with probability 1

We then studied the following quantitative properties:

The minimum probability of reaching a stable configuration within K steps (from any configuration)

\begin{figure}[!h]
  \begin{center}
{\includegraphics[width=0.78\columnwidth]{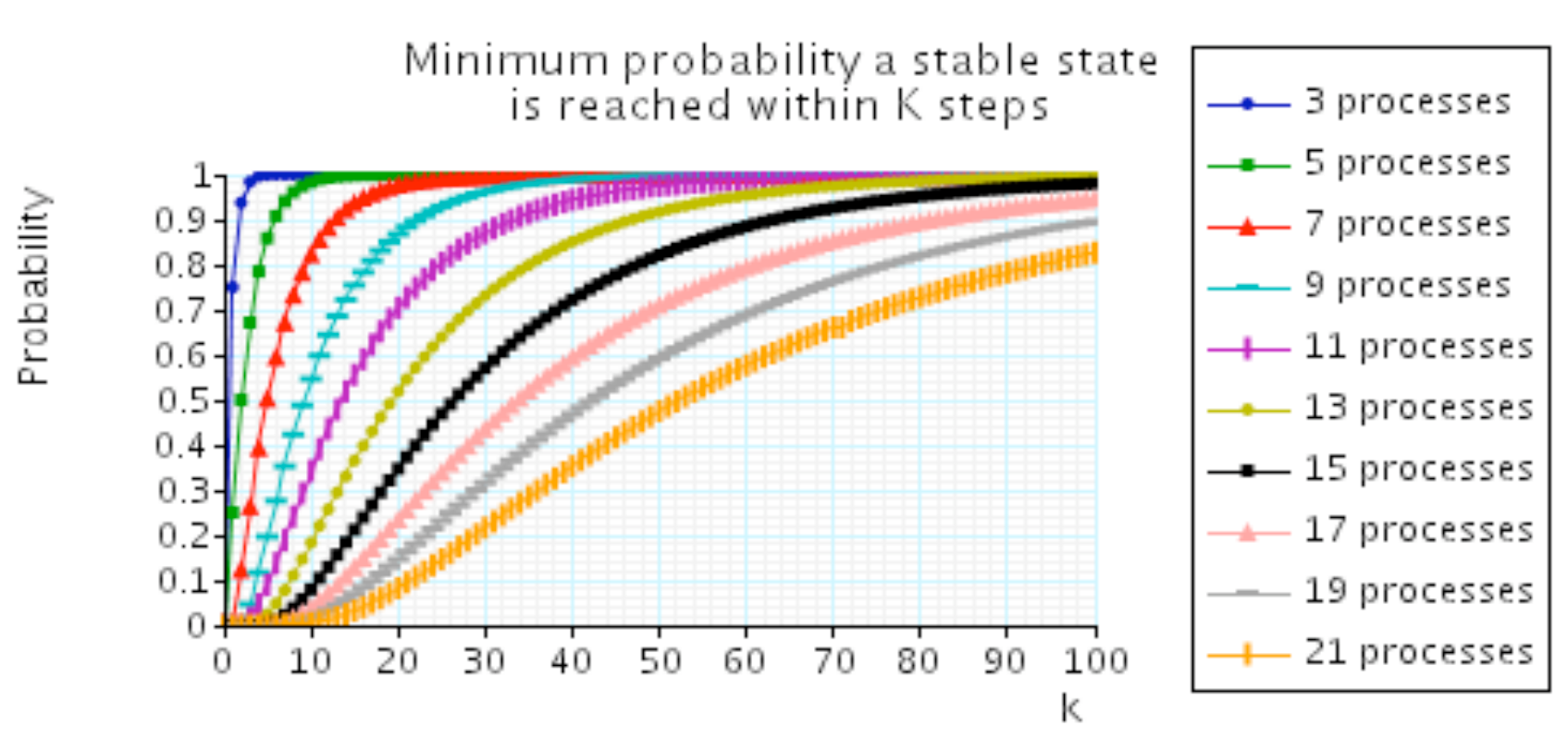}}  \\
{\includegraphics[width=0.78\columnwidth]{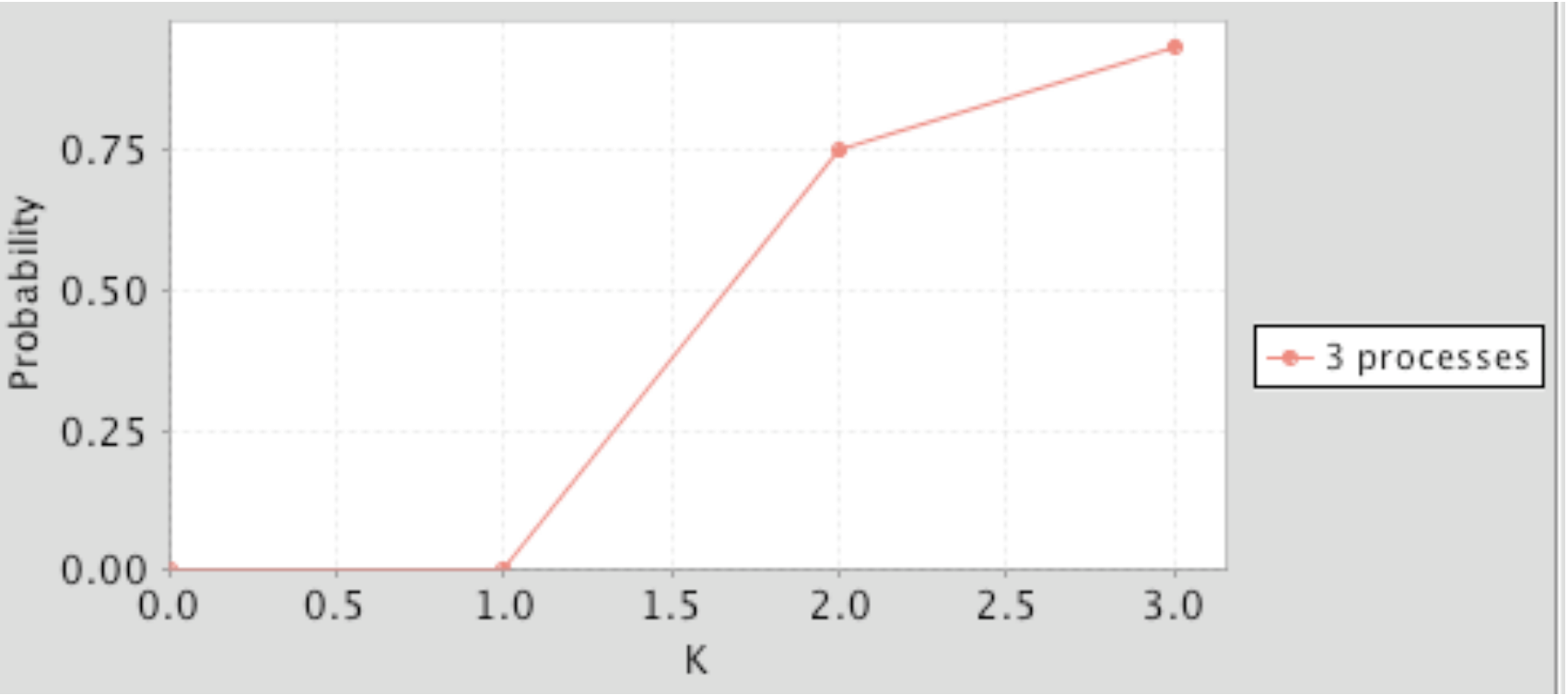}} \ \ 
 
    \caption{The minimum probability of reaching a stable configuration within K steps}
    \label{fig:whenprop}
  \end{center}
\end{figure}

\subsection{zig-zagging problem}
\subsubsection{description}
Zigzagging \cite{free99} is task on which a robot can go either foward, left, or right but (1) it cannot go forward if its preceding action was to go forward, (2) it cannot turn right if its second-last action was to go right, and (3) it cannot turn left if its second-last action was to go left. (frank thesis - page 105)

Valencia models this problem by using cells $a_1$ and $a_2$ to ``look back'' and three different distinct constrats $f,r,l \in D - \{0\}$ and the predicate symbols forward, right, left. \\

GoF \\
GoR \\
GoL \\
Zigzag \\
GoZigZag\\

Valencia verifies that GoZigzag $models$ $square (rombo right ...$.
\subsubsection{simulation}
\paragraph{Frank's style}
\begin{verbatim}
"Ntccrt Simulation" 
"" 
"" 
"Time Unit # 0" 
"Number of processes= 3" 
"" 
"Variables Value (Only those specified are printed)" 
"direction" 
"1" 
"a1" 
"0" 
"a2" 
"0" 
"changea1" 
"1" 
"changea2" 
"1" 
"..........................................." 
"" 
"Time Unit # 1" 
"Number of processes= 9" 
"" 
"Variables Value (Only those specified are printed)" 
"direction" 
"3" 
"a1" 
"1" 
"a2" 
"0" 
"changea1" 
"1" 
"changea2" 
"1" 
"..........................................." 
"" 
"Time Unit # 2" 
"Number of processes= 11" 
"" 
"Variables Value (Only those specified are printed)" 
"direction" 
"1" 
"a1" 
"3" 
"a2" 
"1" 
"changea1" 
"1" 
"changea2" 
"1" 
"..........................................." 
"" 
"Time Unit # 3" 
"Number of processes= 13" 
"" 
"Variables Value (Only those specified are printed)" 
"direction" 
"2" 
"a1" 
"1" 
"a2" 
"3" 
"changea1" 
"1" 
"changea2" 
"1" 
"..........................................." 
"" 
"Time Unit # 4" 
"Number of processes= 15" 
"" 
"Variables Value (Only those specified are printed)" 
"direction" 
"2" 
"a1" 
"2" 
"a2" 
"1" 
"changea1" 
"1" 
"changea2" 
"1" 
"..........................................." 
"" 
"Time Unit # 5" 
"Number of processes= 17" 
"" 
"Variables Value (Only those specified are printed)" 
"direction" 
"1" 
"a1" 
"2" 
"a2" 
"2" 
"changea1" 
"1" 
"changea2" 
"1" 
"..........................................." 
"" 
"Time Unit # 6" 
"Number of processes= 19" 
"" 
"Variables Value (Only those specified are printed)" 
"direction" 
"3" 
"a1" 
"1" 
"a2" 
"2" 
"changea1" 
"1" 
"changea2" 
"1" 
"..........................................." 
"" 
"Time Unit # 7" 
"Number of processes= 21" 
"" 
"Variables Value (Only those specified are printed)" 
"direction" 
"3" 
"a1" 
"3" 
"a2" 
"1" 
"changea1" 
"1" 
"changea2" 
"1" 
"..........................................." 
"" 
"Time Unit # 8" 
"Number of processes= 23" 
"" 
"Variables Value (Only those specified are printed)" 
"direction" 
"1" 
"a1" 
"3" 
"a2" 
"3" 
"changea1" 
"1" 
"changea2" 
"1" 
"..........................................." 
"" 
"Time Unit # 9" 
"Number of processes= 25" 
"" 
"Variables Value (Only those specified are printed)" 
"direction" 
"2" 
"a1" 
"1" 
"a2" 
"3" 
"changea1" 
"1" 
"changea2" 
"1" 
"..........................................." 
"" 
"Time Unit # 10" 
"Number of processes= 27" 
"" 
"Variables Value (Only those specified are printed)" 
"direction" 
"2" 
"a1" 
"2" 
"a2" 
"1" 
"changea1" 
"1" 
"changea2" 
"1" 
"..........................................." 
"" 
"Time Unit # 11" 
"Number of processes= 29" 
"" 
"Variables Value (Only those specified are printed)" 
"direction" 
"3" 
"a1" 
"2" 
"a2" 
"2" 
"changea1" 
"1" 
"changea2" 
"1" 
"..........................................." 
"" 
"Time Unit # 12" 
"Number of processes= 31" 
"" 
"Variables Value (Only those specified are printed)" 
"direction" 
"1" 
"a1" 
"3" 
"a2" 
"2" 
"changea1" 
"1" 
"changea2" 
"1" 
"..........................................." 
\end{verbatim}
\paragraph{My way}
\begin{verbatim}
"Ntccrt Simulation" 
"" 
"SKIP:: cella1 -11::" 
"SKIP:: zigzag -11::" 
"SKIP:: GoR -11::" 
"SKIP:: Exchange value ! -11::" 
"SKIP:: exch0 -11::" 
"SKIP:: I will call exch_aux0 -11::" 
"" 
"Time Unit # 0" 
"Number of processes= 5" 
"" 
"Variables Value (Only those specified are printed)" 
"direction" 
"2" 
"a1" 
"0" 
"a2" 
"0" 
"changea1" 
"1" 
"changea2" 
"1" 
"x" 
"0" 
"y" 
"0" 
"changex" 
"1" 
"changey" 
"[-2147483645..2147483645]" 
"..........................................." 
"SKIP:: zigzag -11::" 
"SKIP:: value changes!!!! -11::" 
"SKIP:: cella1 -11::" 
"SKIP:: GoF -11::" 
"SKIP:: Exchange value ! -11::" 
"SKIP:: exch2 -11::" 
"" 
"Time Unit # 1" 
"Number of processes= 15" 
"" 
"Variables Value (Only those specified are printed)" 
"direction" 
"1" 
"a1" 
"2" 
"a2" 
"0" 
"changea1" 
"1" 
"changea2" 
"1" 
"x" 
"1" 
"y" 
"0" 
"changex" 
"[-2147483645..2147483645]" 
"changey" 
"1" 
"..........................................." 
"SKIP:: zigzag -11::" 
"SKIP:: cella1 -11::" 
"SKIP:: GoL -11::" 
"SKIP:: Exchange value ! -11::" 
"SKIP:: exch1 -11::" 
"" 
"Time Unit # 2" 
"Number of processes= 18" 
"" 
"Variables Value (Only those specified are printed)" 
"direction" 
"3" 
"a1" 
"1" 
"a2" 
"2" 
"changea1" 
"1" 
"changea2" 
"1" 
"x" 
"1" 
"y" 
"1" 
"changex" 
"1" 
"changey" 
"[-2147483645..2147483645]" 
"..........................................." 
"SKIP:: zigzag -11::" 
"SKIP:: cella1 -11::" 
"SKIP:: GoF -11::" 
"SKIP:: Exchange value ! -11::" 
"SKIP:: exch3 -11::" 
"" 
"Time Unit # 3" 
"Number of processes= 21" 
"" 
"Variables Value (Only those specified are printed)" 
"direction" 
"1" 
"a1" 
"3" 
"a2" 
"1" 
"changea1" 
"1" 
"changea2" 
"1" 
"x" 
"0" 
"y" 
"1" 
"changex" 
"[-2147483645..2147483645]" 
"changey" 
"1" 
"..........................................." 
"SKIP:: zigzag -11::" 
"SKIP:: cella1 -11::" 
"SKIP:: GoR -11::" 
"SKIP:: Exchange value ! -11::" 
"SKIP:: exch1 -11::" 
"" 
"Time Unit # 4" 
"Number of processes= 24" 
"" 
"Variables Value (Only those specified are printed)" 
"direction" 
"2" 
"a1" 
"1" 
"a2" 
"3" 
"changea1" 
"1" 
"changea2" 
"1" 
"x" 
"0" 
"y" 
"2" 
"changex" 
"1" 
"changey" 
"[-2147483645..2147483645]" 
"..........................................." 
"SKIP:: zigzag -11::" 
"SKIP:: cella1 -11::" 
"SKIP:: GoR -11::" 
"SKIP:: Exchange value ! -11::" 
"SKIP:: exch2 -11::" 
"" 
"Time Unit # 5" 
"Number of processes= 27" 
"" 
"Variables Value (Only those specified are printed)" 
"direction" 
"2" 
"a1" 
"2" 
"a2" 
"1" 
"changea1" 
"1" 
"changea2" 
"1" 
"x" 
"1" 
"y" 
"2" 
"changex" 
"1" 
"changey" 
"[-2147483645..2147483645]" 
"..........................................." 
"SKIP:: zigzag -11::" 
"SKIP:: cella1 -11::" 
"SKIP:: GoF -11::" 
"SKIP:: Exchange value ! -11::" 
"SKIP:: exch2 -11::" 
"" 
"Time Unit # 6" 
"Number of processes= 30" 
"" 
"Variables Value (Only those specified are printed)" 
"direction" 
"1" 
"a1" 
"2" 
"a2" 
"2" 
"changea1" 
"1" 
"changea2" 
"1" 
"x" 
"2" 
"y" 
"2" 
"changex" 
"[-2147483645..2147483645]" 
"changey" 
"1" 
"..........................................." 
"SKIP:: zigzag -11::" 
"SKIP:: cella1 -11::" 
"SKIP:: GoL -11::" 
"SKIP:: Exchange value ! -11::" 
"SKIP:: exch1 -11::" 
"" 
"Time Unit # 7" 
"Number of processes= 33" 
"" 
"Variables Value (Only those specified are printed)" 
"direction" 
"3" 
"a1" 
"1" 
"a2" 
"2" 
"changea1" 
"1" 
"changea2" 
"1" 
"x" 
"2" 
"y" 
"3" 
"changex" 
"1" 
"changey" 
"[-2147483645..2147483645]" 
\end{verbatim}
\subsubsection{Model checking}
\paragraph{Frank style}
// The robot does not go foward twice
"init" => P<=0 [F store2 = 1 \& store3 = 1]

// The robot does not go right more than twice
"init" => P<=0 [F store2=2 \& store3 = 2 \& store4 = 2]

//The robot does not go left more than twice
"init" => P<=0 [F store2=3 \& store3 = 3 \& store4 = 3] 

// The robot always makes a good move
"init" => P>=1 [F (store2 !=1 | store3 !=1) \& (store2 != 2 | store3 != 2 | store4 != 2) \& (store2 != 3 | store3 != 3 | store4 != 3)]

\paragraph{My way}
Probability distribution for x,y positions

const x;
const y;
P ? [F store7 = x \& store8 = y]

\newcommand{\etalchar}[1]{$^{#1}$}

\end{document}